\documentclass[runningheads]{llncs}

\usepackage[T1]{fontenc}
\usepackage{amsmath,amssymb,mathtools}
\usepackage{microtype}
\usepackage{enumitem}
\usepackage{placeins}
\usepackage[colorlinks,allcolors=blue,bookmarksdepth=2]{hyperref}
\usepackage[nameinlink,capitalize]{cleveref}

\AddToHook{env/proof/end}{\qed}

\crefname{lemma}{Lemma}{Lemmas}
\Crefname{lemma}{Lemma}{Lemmas}
\crefname{theorem}{Theorem}{Theorems}
\Crefname{theorem}{Theorem}{Theorems}

\newcommand{\Z}{\mathbb Z}
\newcommand{\F}{\mathbb F}
\newcommand{\cL}{\mathcal L}
\newcommand{\cC}{\mathcal C}
\newcommand{\cG}{\mathcal G}
\newcommand{\poly}{\operatorname{poly}}
\newcommand{\Ftwo}{F_2}
\newcommand{\Cell}{\operatorname{cell}}
\newcommand{\norm}[1]{\lVert #1\rVert_2}
\newcommand{\mat}[1]{\boldsymbol{#1}}
\newcommand{\vct}[1]{\boldsymbol{#1}}
\newcommand{\mypara}[1]{\textbf{#1.}}

\title{Faster SVP in Polynomial Space}
\author{Yansong Feng \inst{1} \and Yiming Gao \inst{2} \and Jiaqi Liu \inst{1} }

\institute{
Academy of Mathematics and Systems Science,
Chinese Academy of Sciences\\
\email{\{fengyansong,ljqi\}@amss.ac.cn}
\and
School of Cyber Science and Technology,
University of Science and Technology of China\\
\email{qw1234567@mail.ustc.edu.cn}
}

\begin{document}
\maketitle

\begin{abstract}
Kannan's algorithm, as analyzed by Hanrot and Stehl\'e in 2007, solves the exact Euclidean shortest vector problem in polynomial space and \(n^{\frac{n}{2e}+o(n)}\) time. In the classical setting with polynomial space, we obtain the first improvement on this bound via a randomized algorithm that runs in \(n^{\frac{n}{4e}+o(n)}\) time.

The main idea is to represent a fixed shortest vector in many ways as a difference of samples, thereby enabling the low-space collision search of Lyu and Zhu (SODA 2023) to replace exhaustive enumeration in the original analysis.

\end{abstract}

\section{Introduction}
\label{sec:introduction}
A full-rank lattice $\cL\subseteq\mathbb R^n$ is the set of integer
linear combinations of a basis $\mat B=[\vct b_1\ \cdots\ \vct b_n]$.
Given such a basis, the Euclidean shortest vector problem (SVP) asks
for a nonzero vector of minimum norm in $\cL$.  Algorithms for SVP
play an important role in lattice cryptanalysis~\cite{NguyenStern01},
and their cost informs concrete security estimates for lattice-based
cryptography~\cite{Peikert16}.  We study exact SVP in the classical
polynomial-space setting.

Kannan's algorithm is the standard deterministic algorithm in this
setting~\cite{Kannan87}.  It recursively computes
Hermite--Korkine--Zolotarev (HKZ) reductions of projected sublattices
to obtain a quasi-HKZ basis, then enumerates all lattice
vectors of norm at most $\norm{\vct b_1}$.  The enumeration tree can
be traversed in polynomial space~\cite{FinckePohst85,SchnorrEuchner94}.
Kannan bounded the running time by $n^{n+o(n)}$, and Helfrich improved
the analysis to $n^{n/2+o(n)}$~\cite{Kannan87,Helfrich85}.
Hanrot and Stehl\'e obtained the bound
$n^{\frac{n}{2e}+o(n)}$ for the same algorithm~\cite{HanrotStehle07}.
Their analysis exploits the Gram--Schmidt shape of a quasi-HKZ basis
to bound the number of integer points in the enumeration ellipsoids directly,
instead of using an enclosing parallelepiped.  The quasi-HKZ
preprocessing itself is already part of Kannan's recursion. Hanrot and Stehl\'e later constructed HKZ bases on which Kannan's
enumeration procedure requires $n^{\frac{n}{2e}+o(n)}$ bit
operations~\cite{HanrotStehleWorst}.  This matches the leading exponent
of the upper bound for that procedure, but does not rule out a different
search that finds a shortest vector without traversing the entire tree.

A different line of work obtains single-exponential time by allowing
exponential memory. The celebrated AKS sieve proposed by Ajtai, Kumar and Sivakumar gave a randomized
\(2^{O(n)}\)-time and \(2^{O(n)}\)-space algorithm~\cite{AKS01}. Later,
Micciancio and Voulgaris gave a deterministic Voronoi-cell algorithm using
\(\widetilde O(4^n)\) time and \(\widetilde O(2^n)\)
space~\cite{MicciancioVoulgaris13}. Aggarwal, Dadush, Regev, and Stephens-Davidowitz used discrete Gaussian sampling to reduce
the randomized bound to \(2^{n+o(n)}\) time and
space~\cite{ADRS15}. Recently, several improvements to
randomized classical exact-SVP algorithms have yielded progressively smaller
time exponents, culminating in \(2^{n/2+o(n)}\) time, all using
\(2^{n/2+o(n)}\)
space~\cite{GaoFengHu26Super,Hhan26,GaoFengHu26BDGL,Hhan26CosetTree}.
Under heuristic assumptions, Becker, Ducas, Gama and
Laarhoven proposed an algorithm with running time \(2^{0.292n+o(n)}\) and also using
exponential space~\cite{BDGL16}.
These results therefore remain in the exponential-space setting.

Other work seeks to speed up lattice enumeration itself.  Chen et al.\
give a classical time--space tradeoff~\cite{ChenEtAl25}.  At its
balanced point, both time and space are
\(n^{\frac{n}{4e}+o(n)}\).  Aono, Nguyen, and Shen use quantum
enumeration to reach the same time bound with polynomial
space~\cite{AonoNguyenShen18}.
This raises the question:
\begin{center}
    \itshape
    Can SVP be solved classically in
    $n^{\frac{n}{4e}+o(n)}$ time using only polynomial space?
\end{center}

The answer is yes.  We give a randomized classical algorithm that,
given any full-rank integer basis \(\mat B\), returns a shortest nonzero
lattice vector with probability at least \(2/3\).  It uses polynomial
space and runs in \(n^{\frac{n}{4e}+o(n)}\) time up to a
polynomial factor in the input bit length.  The precise statement appears
in \cref{thm:main}.

The key observation concerns Kannan's recursion: it ultimately needs
only one shortest vector, even though enumeration searches the entire
ball.  For a quasi-HKZ basis, we replace this search with the low-space
collision procedure of Lyu and Zhu~\cite{LyuZhu23}.  An exact sampler
provides access to two large random arrays of
lattice points without storing them
explicitly, and the quasi-HKZ shape bound controls their length.  With
constant probability, the two arrays contain many pairs whose difference
is a fixed shortest vector.  For sufficiently many of these pairs,
a random grid assigns the two lattice points in each pair to the same cell.  The
corresponding array entries are then labeled identically and thus form
collisions for the low-space search.  Bounds on repeated samples
and on the number of points in a grid cell control the total number of
collisions.  Repeating the grid construction and collision search
\(2^{O(n)}\) times gives constant success probability.  The resulting
square-root saving over exhaustive enumeration reduces the
exponent from \(1/(2e)\) to \(1/(4e)\).  Finally, embedding this
subroutine in Kannan's recursion extends the method to arbitrary
full-rank integer bases.

Our use of many representations is conceptually reminiscent of the representation technique introduced by Howgrave-Graham and Joux for subset sum~\cite{HGJ10} and further developed by Becker, Coron, and Joux~\cite{BCJ11}. In their subset sum setting, a fixed solution is given many combinatorial representations so that one representation survives suitable modular filters. In our SVP setting, a fixed shortest lattice vector has many representations as a difference of two array entries, and a random grid turns some of these representations into label collisions for the low-space search.

\subsubsection*{Organization.}
\Cref{sec:preliminaries} develops the two tools used throughout: shape
bounds for quasi-HKZ bases and a low-space collision lemma.
\Cref{sec:sampler} uses rounded Gram--Schmidt sampling and a random grid
to construct an implicit label array for the collision search.
\Cref{sec:solver} applies this construction first to a quasi-HKZ basis
and then, through Kannan's recursion, to an arbitrary full-rank lattice basis. \Cref{sec:conclusion} summarizes the result and discusses directions for
future work.

\section{Preliminaries}
\label{sec:preliminaries}
We write \([m]:=\{1,\ldots,m\}\) and use
\(\langle\cdot,\cdot\rangle\) and \(\norm{\cdot}\) for the Euclidean
inner product and norm, respectively.  For a rational matrix \(\mat B\), let
\(|\mat B|\) denote its binary encoding length.  The
symbols \(\log\) and \(\log_2\) denote the natural and binary
logarithms, respectively.  Throughout the recursion, \(n\) is fixed
and \(d\le n\) is the current rank.  Projected bases lie in
\(\mathbb Q^n\) and are scaled uniformly to integer bases when needed.

\subsection{Lattices}

Let \(\mat B=[\vct b_1\ \cdots\ \vct b_d]\in\mathbb R^{n\times d}\) have linearly
independent columns.  The rank-\(d\) lattice generated by the basis
\(\mat B\) is defined as
\[
 \cL(\mat B):=\{\mat B\vct x:\vct x\in\Z^d\}.
\]
We write \(\cL=\cL(\mat B)\) when the basis is clear. The length of its
nonzero shortest vector is denoted by
\[
 \lambda_1(\cL):=
 \min_{\vct v\in\cL\setminus\{0\}}\norm{\vct v}.
\]

\begin{definition}
Given a basis \(\mat B\), the shortest vector problem (SVP) asks for a vector
\(\vct v\in\cL(\mat B)\setminus\{0\}\) satisfying
\(\norm{\vct v}=\lambda_1(\cL(\mat B))\).
\end{definition}

Given $\mat B$, let \(\vct b_1^*,\ldots,\vct b_d^*\) be its
Gram--Schmidt orthogonalization.  Set
\[
 r_i:=\norm{\vct b_i^*}\quad\text{for }i\in[d],
 \]
 and
 \[
 \mu_{j,i}:=\frac{\langle \vct b_j,\vct b_i^*\rangle}
 {\norm{\vct b_i^*}^2}\quad\text{for }i<j.
\]
Then, for \(j\in[d]\),
\[
 \vct b_j=\vct b_j^*+\sum_{i<j}\mu_{j,i}\vct b_i^*.
\]
For \(i\in[d]\), let \(\pi_i\) denote the orthogonal projection onto
\(\operatorname{span}(\vct b_1,\ldots,\vct b_{i-1})^\perp\), with \(\pi_1\) the
identity.  Then \(\vct b_i^*=\pi_i(\vct b_i)\), and
\[
 \pi_i(\cL)
 =\cL\bigl(\pi_i(\vct b_i),\ldots,\pi_i(\vct b_d)\bigr)
\]
is the \(i\)th \emph{projected lattice} of \(\mat B\).

For \(\vct x\in\Z^d\) and \(i\in[d]\), set
\[
 c_i(\vct x):=\sum_{j>i}\mu_{j,i}x_j,
 \qquad y_i(\vct x):=x_i+c_i(\vct x).
 \]
Then \(\mat B\vct x=\sum_i y_i(\vct x)\vct b_i^*\), so \(y_i(\vct x)\) is the
coefficient of \(\vct b_i^*\) in \(\mat B\vct x\).  Orthogonality gives
\[
 \norm{\mat B\vct x}^2=\sum_i r_i^2y_i(\vct x)^2.
\]

The basis \(\mat B\) is \emph{size reduced} if
\(|\mu_{j,i}|\le1/2\) for all \(i<j\).

We next introduce the definitions of HKZ and quasi-HKZ bases~\cite{HanrotStehle07}.

\begin{definition}\label{def:quasi-hkz}
A size-reduced basis is \emph{HKZ reduced} if \(\vct b_i^*\) is a shortest
vector of \(\pi_i(\cL)\) for every \(i\in[d]\).  For \(d\ge2\), a
size-reduced basis is \emph{quasi-HKZ} if \(r_2\ge\norm{\vct b_1}/2\) and the
projected basis \([\pi_2(\vct b_2)\ \cdots\ \pi_2(\vct b_d)]\) is HKZ
reduced.
\end{definition}

The following Hanrot--Stehl\'e product bound will be used
below~\cite[Theorem~3]{HanrotStehle07}.

\begin{lemma}\label{lem:hkz-product}
Let \(\rho_1,\ldots,\rho_m\) be the Gram--Schmidt lengths
of an HKZ-reduced basis of rank \(m\).  For every
nonempty \(I\subseteq[m]\), writing \(k:=|I|\),
\[
 \frac{\rho_1^k}{\prod_{i\in I}\rho_i}\le m^{\frac{k}{2}(1+\log(m/k))}\le m^{\frac{m}{2e}+k/2}.
\]
\end{lemma}

For the shape bounds below, assume \(d\ge2\).  Write
\(\lambda_1:=\lambda_1(\cL)\), and set
\[
 s_i:=\max\left\{1,\frac{\norm{\vct b_1}}{\sqrt d\,r_i}\right\}
 \quad(i\in[d]),\qquad P:=\prod_{i=1}^d s_i.
\]

Now we introduce the shape bounds for quasi-HKZ bases in the following \cref{lem:hkz-shape}, which bounds these scale parameters and compares \(\norm{\vct b_1}\) with
\(\lambda_1\).

\begin{lemma}\label{lem:hkz-shape}
For a quasi-HKZ basis of rank \(d\ge2\),
\[
 P\le 2^{O(d)}d^{\frac{d}{2e}}=d^{\frac{d}{2e}+o(d)},
 \qquad \max_{i\in[d]}s_i=d^{o(d)},
 \qquad \norm{\vct b_1}\le2\lambda_1.
\]
\end{lemma}

\begin{proof}
Let \(S:=\{i:s_i>1\}\).  Since \(s_1=1\), we have
\(S\subseteq\{2,\ldots,d\}\).  If \(S=\varnothing\), then \(P=1\).
Otherwise, set \(k:=|S|\).  Applying \cref{lem:hkz-product}
with \(m=d-1\) and \(\rho_j=r_{j+1}\) to the HKZ-reduced
basis \([\pi_2(\vct b_2)\ \cdots\ \pi_2(\vct b_d)]\) and the index set
\(\{i-1:i\in S\}\) gives
\[
 \frac{r_2^k}{\prod_{i\in S}r_i}
 \le (d-1)^{(d-1)/(2e)+k/2}.
\]
Using \(\norm{\vct b_1}\le2r_2\), we obtain
\begin{align*}
 P&=\frac{\norm{\vct b_1}^k}{d^{k/2}\prod_{i\in S}r_i}\\
  &\le \frac{2^k}{d^{k/2}}
       (d-1)^{(d-1)/(2e)+k/2}
   \le 2^d d^{\frac{d}{2e}}.
\end{align*}
For each \(i\in\{2,\ldots,d\}\), the same lemma applied to the
projected basis with \(I=\{i-1\}\) gives
\[
 \frac{r_2}{r_i}\le(d-1)^{\frac12(1+\log(d-1))}.
\]
Together with \(\norm{\vct b_1}\le2r_2\) and \(s_1=1\), this yields
\[
 \max_{i\in[d]}s_i
 \le2d^{(\log d)/2}=d^{O(\log d)}=d^{o(d)}.
\]

Finally, let \(0\ne\vct v\in\cL\).  If \(\vct v\) is parallel to \(\vct b_1\),
then \(\vct v=t\vct b_1\) for some \(t\in\Z\setminus\{0\}\), so
\(\norm{\vct v}\ge\norm{\vct b_1}\).  Otherwise,
\(\pi_2(\vct v)\in\pi_2(\cL)\setminus\{0\}\).  Since \(\vct b_2^*\) is shortest in
\(\pi_2(\cL)\),
\[
 \norm{\vct v}\ge\norm{\pi_2(\vct v)}\ge r_2\ge\norm{\vct b_1}/2.
\]
Thus \(\lambda_1\ge\norm{\vct b_1}/2\), which proves \(\norm{\vct b_1}\le2\lambda_1\).
\end{proof}

\subsection{Low-space collision search}

Our algorithm uses a low-space collision procedure on an implicitly
represented array.  This idea goes back to Floyd's
cycle-finding algorithm~\cite{Knuth69} and its use in the
Pollard--\(\rho\) method~\cite{Pollard75} to find a collision among
successive iterates.  Beame, Clifford, and Machmouchi~\cite{BCM13} adapted
this approach to array collisions in the random-oracle model.
Chen, Jin, Williams, and Wu~\cite{CJWW22} instantiated this procedure
with a short-seed pseudorandom hash family, thereby removing its
random-oracle requirement in the polylogarithmic-space regime.  Lyu
and Zhu~\cite{LyuZhu23} subsequently simplified and
strengthened the analysis.

To define a collision walk, fix an array \(\vct u\in[N]^L\) and a hash
\(h:[N]\to[L]\cup\{\bot\}\), and set
\[
 f_h(i):=h(u_i)\qquad(i\in[L]).
\]
Starting from \(x_0:=s\in[L]\), iterate
\[
 x_0\longrightarrow x_1=f_h(x_0)
 \longrightarrow x_2=f_h(x_1)\longrightarrow\cdots .
\]
The walk stops at \(\bot\) or the first repeated index.
The set of visited indices is the \emph{orbit} of the walk,
denoted by \(\mathcal O(h,s)\).  If
\(u_i=u_j\) for distinct \(i,j\in\mathcal O(h,s)\), then
\(f_h(i)=f_h(j)\).  Because the walk stops at the first repeated
index, \(\{i,j\}\) is the only pair of visited indices with the same
successor.  A cycle finder~\cite[Theorem~2.1]{BCM13}
recovers this pair using \(O(|\mathcal O(h,s)|)\) array accesses and
\(O(\log L)\) additional bits, without storing the orbit.  It outputs
the recovered pair only if its two array values agree.

The \emph{second frequency moment} of \(\vct u\) is
\[
 \Ftwo(\vct u):=\sum_{a\in[N]}|\{i\in[L]:u_i=a\}|^2.
\]

We use the following orbit bounds of
Lyu and Zhu~\cite[Lemmas~3 and~4]{LyuZhu23}.

\begin{lemma}\label{lem:lyu-zhu-orbit}
Let \(\vct u\in[N]^L\), where \(N=L^{O(1)}\).  There is a distribution
on hashes \(h:[N]\to[L]\cup\{\bot\}\) with \(O(\log^3 L)\)-bit seeds
and \(\poly(\log L)\)-time evaluation.  Let \(h\) be
sampled from this distribution, and let \(s\) be chosen independently
and uniformly from \([L]\).
For distinct \(i,j\in[L]\),
\begin{align*}
 \Pr_{h,s}[i\in\mathcal O(h,s)]
   &=O(1/\sqrt L),\\
 \Pr_{h,s}[\{i,j\}\subseteq\mathcal O(h,s)]
   &=\Omega\bigl(1/\Ftwo(\vct u)\bigr).
\end{align*}
\end{lemma}

A \emph{duplicate pair} is an index pair \(\{i,j\}\) with
\(i\ne j\) and \(u_i=u_j\).  The orbit bounds control the expected
length of a walk and its chance of finding a fixed duplicate pair.
Repeating the walk gives a search bound for any set \(\cG\) of such pairs.

\begin{lemma}\label{lem:collision-search}
Let \(\vct u\) be a random-access array of length \(L\) over a universe of
size \(L^{O(1)}\).  Assume that
\begin{enumerate}[label=(\roman*),nosep]
\item \(\Ftwo(\vct u)\le p\le L^2\);
\item \(\cG\) is a set of duplicate pairs with \(|\cG|\ge r_0\ge1\).
\end{enumerate}
There is a randomized algorithm that, given random access to \(\vct u\)
and the bounds \(p,r_0\), outputs a sequence of duplicate pairs.
The sequence contains a pair in \(\cG\) with constant probability.
On every run, the algorithm makes
\[
 \widetilde O\left(\sqrt L\,p/r_0\right)
\]
array accesses and uses \(O(\log^3 L)\) working bits in addition
to any space needed to evaluate an entry \(u_i\).
\end{lemma}

\begin{proof}
Sample \(h\) and \(s\) as in \cref{lem:lyu-zhu-orbit}
and run the cycle finder described above.
Summing the first bound of
\cref{lem:lyu-zhu-orbit} over \(i\in[L]\) gives
\[
 \mathbb E_{h,s}[|\mathcal O(h,s)|]=O(\sqrt L).
\]
Thus each trial uses \(O(\sqrt L)\) expected array accesses.
By the second bound, each
fixed duplicate pair is reached and thus output with probability
\(\Omega(1/\Ftwo(\vct u))\).  A single orbit contains at most one
duplicate pair, so the output events for pairs in \(\cG\) are disjoint.
Thus one trial outputs a pair in \(\cG\) with probability
\[
 \Omega\bigl(|\cG|/\Ftwo(\vct u)\bigr)\ge\Omega(r_0/p).
\]

Run \(\Theta(p/r_0)\) independent trials and output every duplicate
pair found.  The resulting sequence contains a pair in \(\cG\) with
constant probability, and its expected total array-access count is
\(O(\sqrt L\,p/r_0)\).  Let \(A\) be the total array-access count
when all these trials run to completion.  Let \(\eta>0\) be a success
lower bound.  For
\(T:=C\sqrt L\,p/r_0\) with sufficiently large \(C\), Markov's inequality
gives
\[
 \Pr[A>T]\le\frac{\mathbb E[A]}{T}\le\eta/2.
\]
Thus, truncation at \(T\) accesses yields success probability
at least \(\eta/2\) and a worst-case access bound \(T\).  The trial
counter uses \(O(\log L)\) bits since
\(p/r_0\le L^2\), and the hash seed and orbit state use
\(O(\log^3 L)\) bits.
\end{proof}

The walk may also output pairs outside \(\cG\).  We may check
each output pair with a polynomial-space predicate and retain the best
accepted pair, while continuing the prescribed trials.  The collision
search needs only \(\vct u,p,r_0\).  The set \(\cG\) is
used only in the success analysis.

\section{Sampling and grid collisions}
\label{sec:sampler}
To apply the collision search of
\cref{lem:collision-search} to a quasi-HKZ basis, we construct two
implicit arrays of lattice samples and assign a short label to each
sample.  The collision search is run on their concatenated label array:
its \(t\)th entry is the label of the \(t\)th sample and is
computed only when queried.  We need many equal-label pairs whose
underlying samples differ by a fixed shortest vector.  At the same time, the label
array's second frequency moment, which counts ordered pairs of positions
with equal labels, must remain small.  This section gives such a
construction.

Fix a quasi-HKZ integer basis \(\mat B\) of rank \(d\ge2\), write
\(\cL:=\cL(\mat B)\) and
\(\lambda_1:=\lambda_1(\cL)\), and let \(s_i,P\) be the shape parameters from
\cref{sec:preliminaries}.  Our starting point is a finitely supported
distribution \(\nu\) on \(\cL\) with two properties.  First, every lattice
point should have small probability.  Second, \(\nu\) should
have enough overlap with its translate by every lattice vector
\(\mat B\vct z\) satisfying
\(\norm{\mat B\vct z}\le\norm{\vct b_1}\).  This range contains every
shortest vector because \(\lambda_1\le\norm{\vct b_1}\).  For
independent samples \(\vct V\) and \(\vct W\) drawn from \(\nu\), the overlap is
\[
 \sum_{\vct u\in\cL}\nu(\vct u)\nu(\vct u-\mat B\vct z)
 =\Pr[\vct V-\vct W=\mat B\vct z].
\]
The overlap bound gives many pairs, with one sample from each of two
arrays, whose difference is a shortest vector.  The point-mass bound
limits repeated sample values.

The construction proceeds in four steps.  First,
\cref{subsec:rounded-gs} introduces rounded Gram--Schmidt
coordinates and bounds their change under a short lattice translation
in a weighted norm.  Second, \cref{subsec:sampling-distribution}
defines a finite product distribution on these coordinates that can be
sampled exactly and has enough overlap with every short translate.
Third, \cref{subsec:representations} uses short random seeds to define
two implicit sample arrays.  The two distribution bounds give many
pairs with the desired difference while controlling repeated sample
values.  Finally, \cref{subsec:grid-collisions} uses a random grid to
assign both points of many target pairs the same cell label.  A packing argument
converts the bound on repeated sample values into a bound on the label
array's second frequency moment, and a final hash compresses the labels into
the range required by \cref{lem:collision-search}.

\subsection{Rounded Gram--Schmidt coordinates}
\label{subsec:rounded-gs}

We introduce rounded Gram--Schmidt coordinates and prove that they
define a bijection of \(\Z^d\).  Unlike the original coefficient
coordinates, whose change need not be controlled by
\(\norm{\mat B\vct z}\), these coordinates give a bounded weighted
coordinate change for every short lattice translation.  This bound
underlies the overlap analysis below.

Define \(\lfloor t\rceil:=\lfloor t+1/2\rfloor\).  Then
\(\lfloor t+k\rceil=\lfloor t\rceil+k\) for \(t\in\mathbb R\) and
\(k\in\Z\).  For \(\vct x\in\Z^d\) and \(i\in[d]\), define
\[
\begin{aligned}
 \kappa_i(\vct x)&:=x_i+\lfloor c_i(\vct x)\rceil,\\
 \kappa(\vct x)&:=(\kappa_1(\vct x),\ldots,\kappa_d(\vct x)).
\end{aligned}
\]
Thus \(\kappa_i(\vct x)=\lfloor y_i(\vct x)\rceil\).  We call
\(\kappa(\vct x)\) the \emph{rounded Gram--Schmidt coordinates}
of \(\mat B\vct x\).

We round each scale \(s_i\) up to a power of two so that the coordinate
weights can be computed exactly.  Let \(a_i\) be the least power of two
satisfying
\[
 a_i\ge1,\qquad
 a_i^2 d r_i^2\ge \norm{\vct b_1}^2.
\]
Set
\[
 \alpha_i:=a_i^{-2},\qquad \Pi:=\prod_i a_i,\qquad Q:=\max_i a_i.
\]
These quantities are exactly computable from the rational squared
Gram--Schmidt lengths, and they satisfy
\begin{equation}\label{eq:rounded-scales}
s_i\le a_i<2s_i,\qquad
 P\le\Pi<2^dP,\qquad
Q<2\max_{i\in[d]}s_i=d^{o(d)},
\end{equation}
and
\begin{equation}\label{eq:alpha-domination}
 \alpha_i\le\min\left\{1,\frac{d r_i^2}{\norm{\vct b_1}^2}\right\}.
\end{equation}
The next lemma proves that \(\kappa\) is bijective and bounds the
coordinate change induced by a short lattice translation.

\begin{lemma}\label{lem:rounded-gs}
The map \(\kappa:\Z^d\to\Z^d\) is a bijection.  For fixed
\(\vct z\in\Z^d\), define \(T_{\vct z}:\Z^d\to\Z^d\) and its
coordinate changes by
\[
\begin{aligned}
 T_{\vct z}(\vct k)&:=\kappa(\kappa^{-1}(\vct k)+\vct z),\\
 \delta_i(\vct k)&:=T_{\vct z}(\vct k)_i-k_i
 \qquad(\vct k\in\Z^d,\ i\in[d]).
\end{aligned}
\]
For every \(\vct k\in\Z^d\) and \(i\in[d]\),
\[
 |\delta_i(\vct k)-y_i(\vct z)|\le1.
\]
If \(\norm{\mat B\vct z}\le \norm{\vct b_1}\), then
\[
 \sum_i\alpha_i\delta_i(\vct k)^2\le4d.
\]
\end{lemma}

\begin{proof}
Fix \(\vct k\in\Z^d\).  Since \(c_i(\vct x)\) depends only on
\(x_{i+1},\ldots,x_d\), define \(\vct x\) by backward substitution:
\[
 x_i:=k_i-\left\lfloor
 \sum_{j=i+1}^d\mu_{j,i}x_j\right\rceil,
 \qquad i=d,d-1,\ldots,1.
\]
The resulting \(\vct x\in\Z^d\) satisfies \(\kappa(\vct x)=\vct k\), so \(\kappa\)
is surjective.  Every preimage of \(\vct k\) must satisfy the same recursion,
so it is unique and \(\kappa\) is also injective.

Let \(\vct x\) denote this preimage and set
\(\theta(t):=\lfloor t\rceil-t\), so
\(|\theta(t)|\le1/2\).
Using \(c_i(\vct x+\vct z)=c_i(\vct x)+c_i(\vct z)\) and
\(y_i(\vct z)=z_i+c_i(\vct z)\), we obtain
\[
\begin{aligned}
 \delta_i(\vct k)-y_i(\vct z)
 &=\kappa_i(\vct x+\vct z)-\kappa_i(\vct x)-y_i(\vct z)\\
 &=\lfloor c_i(\vct x)+c_i(\vct z)\rceil
    -\lfloor c_i(\vct x)\rceil-c_i(\vct z)\\
 &=\theta(c_i(\vct x)+c_i(\vct z))-\theta(c_i(\vct x)).
\end{aligned}
\]
Thus \(|\delta_i(\vct k)-y_i(\vct z)|\le1\).

Now suppose that \(\norm{\mat B\vct z}\le \norm{\vct b_1}\).  By
Equation~\eqref{eq:alpha-domination},
\[
\begin{aligned}
 \sum_i\alpha_i y_i(\vct z)^2
 &\le\frac{d}{\norm{\vct b_1}^2}\sum_i r_i^2y_i(\vct z)^2\\
 &=\frac{d}{\norm{\vct b_1}^2}\norm{\mat B\vct z}^2
 \le d.
\end{aligned}
\]
Also, \(\sum_i\alpha_i\le d\).  Combining these estimates with
\(|\delta_i(\vct k)-y_i(\vct z)|\le1\) gives
\[
\begin{aligned}
 \sum_i\alpha_i\delta_i(\vct k)^2
 &\le2\sum_i\alpha_i y_i(\vct z)^2
     +2\sum_i\alpha_i\bigl(\delta_i(\vct k)-y_i(\vct z)\bigr)^2\\
 &\le2d+2\sum_i\alpha_i
 \le4d.
\end{aligned}
\]
\end{proof}

\subsection{Constructing the sampling distribution}
\label{subsec:sampling-distribution}
We now construct \(\nu\) as a product distribution in the rounded
Gram--Schmidt coordinates, using Gaussian-like weights in each
coordinate.  By \cref{lem:rounded-gs}, translation by a short lattice
vector changes these coordinates by only a bounded amount in the
weighted norm, so the product distribution keeps enough overlap with
every such translate.  We truncate each coordinate distribution to
obtain finite support.  We then approximate each one by a distribution
that can be sampled exactly from a fixed number of random bits.  The
resulting distribution \(\nu\) has a small maximum point mass and a
uniform lower bound on its overlap with every short translate.
\Cref{thm:sampler-autocorrelation} gives the precise bounds, which will be used to
analyze the sample arrays in \cref{subsec:representations}.

For coordinate \(i\), truncate the support and round the weight
\(2^{-\alpha_i k^2}\) down to an integer power of two.  Set
\[
 I_i:=\{k\in\Z:\alpha_i k^2\le16d\},\qquad
 e_i(k):=\lceil\alpha_i k^2\rceil,\qquad
 w_i(k):=2^{-e_i(k)}.
\]
The set \(I_i\) is the finite support for coordinate \(i\).  Normalize
the weights by setting
\[
 Z_i:=\sum_{k\in I_i}w_i(k),\qquad
 p_i(k):=\frac{w_i(k)}{Z_i}\quad(k\in I_i).
\]

We first bound \(Z_i\), \(|I_i|\), and the smallest point mass of
\(p_i\).

\begin{lemma}\label{lem:coordinate-normalizer}
For each \(i\in[d]\),
\[
 \frac{a_i}{4}\le Z_i\le4a_i,\qquad
 |I_i|\le9\sqrt d\,a_i,\qquad
 \min_{k\in I_i}p_i(k)\ge\frac{2^{-16d}}{4a_i}.
\]
\end{lemma}

\begin{proof}
Since \(\alpha_i=a_i^{-2}\), a direct count gives
\[
 |I_i|=2\lfloor4\sqrt d\,a_i\rfloor+1\le9\sqrt d\,a_i.
\]
Moreover, \(e_i(k)\le16d\) and \(w_i(k)\ge2^{-16d}\) for
\(k\in I_i\).  Since \(w_i(k)\le2^{-k^2/a_i^2}\), comparison with the
Gaussian integral gives
\[
 Z_i\le\sum_{k\in\Z}2^{-k^2/a_i^2}
 \le1+a_i\sqrt{\pi/\log2}<4a_i.
\]
For the lower bound, the fact that \(a_i\) is a power of two gives
\[
 \bigl|\{k\in\Z:|k|\le a_i/4\}\bigr|
 =2\lfloor a_i/4\rfloor+1\ge a_i/2.
\]
For each such \(k\), we have \(\alpha_i k^2\le1/16\), so
\(k\in I_i\), \(e_i(k)\le1\), and \(w_i(k)\ge1/2\).  Hence
\[
 Z_i\ge\sum_{|k|\le a_i/4}w_i(k)\ge\frac{a_i}{4}.
\]
The lower bound on \(p_i(k)\) follows from \(w_i(k)\ge2^{-16d}\)
and \(Z_i\le4a_i\).
\end{proof}

We next round the probabilities of \(p_i\) so that a coordinate can be
sampled from a fixed-length random bit string.  The construction assigns
an integer number of bit strings to each \(k\in I_i\), while preserving
all point masses within a constant factor.  Set
\[
\begin{aligned}
 \omega_i(k)&:=2^{16d}w_i(k)=2^{16d-e_i(k)}\in\Z_{\ge1}
 \quad(k\in I_i),\\
 \Omega_i&:=\sum_{k\in I_i}\omega_i(k)=2^{16d}Z_i,
 \qquad
 \ell_i:=\lceil\log_2(2\Omega_i)\rceil.
\end{aligned}
\]
For each \(k\in I_i\), first set
\[
 \widetilde m_i(k)
 :=\left\lfloor\frac{2^{\ell_i}\omega_i(k)}{\Omega_i}\right\rfloor,
 \qquad
 D_i:=2^{\ell_i}-\sum_{k\in I_i}\widetilde m_i(k).
\]
As a sum of \(|I_i|\) fractional parts,
\(0\le D_i<|I_i|\).  It counts the random inputs left unassigned by
the rounded-down counts.  Thus
the following adjustment is well defined.  Increase
\(\widetilde m_i(k)\) by one for the first \(D_i\) elements of \(I_i\) in
increasing order.  Denote the resulting counts by \(m_i(k)\).  The
rounded probabilities are
\[
 \widehat p_i(k):=\frac{m_i(k)}{2^{\ell_i}}
 \qquad(k\in I_i).
\]

The next lemma gives the pointwise approximation guarantee for
\(\widehat p_i(k)\) and the complexity of sampling from \(\widehat p_i\).

\begin{lemma}\label{lem:coordinate-sampler}
For each \(i\in[d]\), the distribution \(\widehat p_i\) satisfies
\[
 \frac12p_i(k)\le\widehat p_i(k)\le\frac32p_i(k)
 \qquad(k\in I_i).
\]
A sample from \(\widehat p_i\) uses \(\ell_i=O(d+\log a_i)\)
random bits.  It can be generated with two
passes over \(I_i\) in \(d^{o(d)}\poly(|\mat B|)\) time and polynomial space.
\end{lemma}

\begin{proof}
Since \(\omega_i(k)/\Omega_i=p_i(k)\), the construction gives
\[
\begin{gathered}
 \sum_{k\in I_i}m_i(k)=2^{\ell_i},\qquad
 \widetilde m_i(k)=\lfloor2^{\ell_i}p_i(k)\rfloor,\\
 m_i(k)\in\{\widetilde m_i(k),\widetilde m_i(k)+1\}.
\end{gathered}
\]
Thus \(\lvert m_i(k)-2^{\ell_i}p_i(k)\rvert\le1\).  Also,
\(2^{\ell_i}\ge2\Omega_i\) and \(\omega_i(k)\ge1\) give
\[
 2^{\ell_i}p_i(k)
 =\frac{2^{\ell_i}\omega_i(k)}{\Omega_i}\ge2.
\]
Therefore
\[
 \lvert\widehat p_i(k)-p_i(k)\rvert
 =2^{-\ell_i}\lvert m_i(k)-2^{\ell_i}p_i(k)\rvert
 \le2^{-\ell_i}\le\frac12p_i(k),
\]
which proves the pointwise bounds.

For \(e\in\{0,\ldots,16d\}\), let
\[
 N_{i,e}:=|\{k\in I_i:e_i(k)=e\}|.
\]
A first pass over \(k\in I_i\) computes these \(16d+1\)
multiplicities, and hence
\[
\begin{aligned}
 \Omega_i&=\sum_{e=0}^{16d}N_{i,e}2^{16d-e},\\
 D_i&=2^{\ell_i}-\sum_{e=0}^{16d}N_{i,e}
 \left\lfloor\frac{2^{\ell_i}2^{16d-e}}{\Omega_i}\right\rfloor.
\end{aligned}
\]
Choose \(U\) uniformly from
\(\{0,\ldots,2^{\ell_i}-1\}\).  During a second pass over
\(k\in I_i\) in increasing order, output the least \(k\in I_i\) such that
\[
 U<\sum_{u\in I_i,\,u\le k}m_i(u).
\]
Exactly \(m_i(k)\) values of \(U\) produce \(k\), so the output has
distribution \(\widehat p_i\).

By \cref{lem:coordinate-normalizer},
\[
 \Omega_i=2^{16d}Z_i\le2^{16d+2}a_i,
 \qquad \ell_i=O(d+\log a_i).
\]
Each pass has \(|I_i|\le9\sqrt d\,Q=d^{o(d)}\) iterations by
Equation~\eqref{eq:rounded-scales}.  Each iteration uses \(O(1)\) exact
arithmetic operations on integers of \(O(\ell_i+\log a_i)\) bits.  Since
\(\log a_i\le\log Q=o(d\log d)\) and \(d\le|\mat B|\), these operations cost
\(\poly(|\mat B|)\) bit operations.  Thus the two passes take
\(d^{o(d)}\poly(|\mat B|)\) time.  The values \(N_{i,e}\) and
the cumulative sum use \(O(d\ell_i)\) bits.  All remaining stored data
use \(\poly(|\mat B|)\) bits.
\end{proof}

Let \(\widehat\nu:=\bigotimes_{i=1}^d\widehat p_i\) be the
product distribution on \(\prod_i I_i\).  Define the bijection
\(\Phi:\Z^d\to\cL\) and the induced distribution \(\nu\) by
\[
 \Phi(\vct k):=\mat B\kappa^{-1}(\vct k),\qquad
 \nu(\Phi(\vct k)):=\widehat\nu(\vct k)=\prod_i\widehat p_i(k_i).
\]
Their supports are
\[
 \operatorname{supp}(\widehat\nu)=\prod_i I_i,
 \qquad \operatorname{supp}(\nu)=\Phi\Bigl(\prod_i I_i\Bigr).
\]

The following theorem bounds the point masses of \(\nu\) and its
overlap under every short lattice translation.

\begin{theorem}
\label{thm:sampler-autocorrelation}
The distribution \(\nu\) satisfies
\[
 \|\nu\|_\infty\le\frac{2^{3d}}{\Pi}.
\]
Let \(\vct V\) and \(\vct W\) be independent samples from \(\nu\).
For every \(\vct z\in\Z^d\) satisfying
\(\norm{\mat B\vct z}\le\norm{\vct b_1}\),
\[
 \rho_{\vct z}:=\Pr[\vct V-\vct W=\mat B\vct z]
 =\sum_{\vct v\in\cL}\nu(\vct v)\nu(\vct v-\mat B\vct z)
 \ge\frac{2^{-20d}}{\Pi}.
\]
\end{theorem}

\begin{proof}
By \cref{lem:coordinate-normalizer}, \(Z_i\ge a_i/4\).  Since
\(w_i(k)\le1\), \(p_i(k)=w_i(k)/Z_i\le4/a_i\) for \(k\in I_i\).  By
\cref{lem:coordinate-sampler}, \(\widehat p_i(k)\le(3/2)p_i(k)\le6/a_i\).
Since \(\Phi\) is a bijection,
\[
 \|\nu\|_\infty=\|\widehat\nu\|_\infty
 \le\frac{6^d}{\Pi}\le\frac{2^{3d}}{\Pi}.
\]

Fix \(\vct z\in\Z^d\) with \(\norm{\mat B\vct z}\le \norm{\vct b_1}\), and set
\[
\begin{aligned}
 \cC&:=\{\vct k\in\Z^d:|k_i|\le a_i/4\text{ for all }i\},\\
 |\cC|&=\prod_{i=1}^d\bigl(2\lfloor a_i/4\rfloor+1\bigr)
 \ge\frac{\Pi}{2^d}\ge\frac{\Pi}{4^d}.
\end{aligned}
\]
For \(\vct k\in\cC\), \(\alpha_i k_i^2\le1/16\) for every \(i\).  Hence
\(\cC\subseteq\prod_i I_i\) and
\(\sum_i\alpha_i k_i^2\le d/16\).  Let
\(\vct k':=T_{\vct z}(\vct k)\).
Writing \(\delta_i(\vct k)=k_i'-k_i\), \cref{lem:rounded-gs} gives
\(\sum_i\alpha_i\delta_i(\vct k)^2\le4d\).  Therefore
\[
 \sum_i\alpha_i(k_i')^2
 \le2\sum_i\alpha_i k_i^2+2\sum_i\alpha_i\delta_i(\vct k)^2
 \le\frac d8+8d=\frac{65d}{8}<16d.
\]
Thus \(\vct k'=T_{\vct z}(\vct k)\in\prod_{i=1}^d I_i\).

Since \(e_i(t)=\lceil\alpha_i t^2\rceil\le\alpha_i t^2+1\),
\[
\begin{aligned}
 \sum_i e_i(k_i)
 &\le\sum_i\alpha_i k_i^2+d
 \le\frac{17d}{16}<2d,\\
 \sum_i e_i(k_i')
 &\le\sum_i\alpha_i(k_i')^2+d
 \le\frac{73d}{8}<10d.
\end{aligned}
\]
By \cref{lem:coordinate-normalizer,lem:coordinate-sampler}, for
\(t\in I_i\),
\[
 \widehat p_i(t)\ge\frac{p_i(t)}{2}
 =\frac{2^{-e_i(t)}}{2Z_i}\ge\frac{2^{-e_i(t)}}{8a_i}.
\]
Consequently,
\[
\begin{aligned}
 \widehat\nu(\vct k)
 &\ge\frac{2^{-3d-\sum_i e_i(k_i)}}{\Pi}
 \ge\frac{2^{-5d}}{\Pi},\\
 \widehat\nu(\vct k')
 &\ge\frac{2^{-3d-\sum_i e_i(k_i')}}{\Pi}
 \ge\frac{2^{-13d}}{\Pi}.
\end{aligned}
\]
Since \(\kappa\) and \(\vct x\mapsto\vct x+\vct z\) are
bijections, \(T_{\vct z}\) is bijective.  Moreover,
\(\Phi(T_{\vct z}(\vct k))-\Phi(\vct k)=\mat B\vct z\).  Hence
\[
 \rho_{\vct z}
 \ge\sum_{\vct k\in\cC}\widehat\nu(\vct k)\widehat\nu(T_{\vct z}(\vct k))
 \ge\frac{\Pi}{4^d}
      \frac{2^{-5d}}{\Pi}\frac{2^{-13d}}{\Pi}
 =\frac{2^{-20d}}{\Pi}.
\]
\end{proof}

Generating one lattice point according to the distribution \(\nu\) uses
\[
 \ell:=\sum_i\ell_i=O(d^2+\log\Pi)
\]
random bits.  Split these bits into consecutive blocks of lengths
\(\ell_1,\ldots,\ell_d\), and use the \(i\)th block to generate
\(k_i\) according to \(\widehat p_i\) as in
\cref{lem:coordinate-sampler}.  For
\(\vct k=(k_1,\ldots,k_d)\), recover \(\vct x=\kappa^{-1}(\vct k)\) by
backward substitution.
Size reduction gives
\[
 |x_i|\le|k_i|+\frac12\sum_{j>i}|x_j|+\frac12,
\]
so each coefficient has \(O(d+\log Q+\log d)\) bits.  Combining
\cref{lem:coordinate-sampler} with this coefficient bound, computing the
resulting lattice point \(\Phi(\vct k)\) takes
\(d^{o(d)}\poly(|\mat B|)\) time and polynomial space.

\subsection{Many representations of a shortest vector}
\label{subsec:representations}
We now turn the two bounds on \(\nu\) into high-probability properties
of two sample arrays \(\vct X\) and \(\vct Y\).  For the analysis, fix a
shortest vector \(\vct v^*\).  The arrays are generated without knowing
\(\vct v^*\).  The overlap lower bound makes the number \(Z\) of pairs
satisfying \(\vct X_i-\vct Y_j=\vct v^*\) large with high probability.
The maximum point-mass bound limits repeated sample values and hence
controls \(F_{\rm pt}\), the second frequency moment of the concatenated
sample array.  \Cref{thm:many-representations} gives the lower bound on
\(Z\), and \cref{thm:point-frequency} gives the upper bound on
\(F_{\rm pt}\).  We then show how a short random seed for each array
lets any entry be recomputed on demand.

Set
\begin{equation}\label{eq:implicit-array-length}
 m:=2^{64n}\Pi.
\end{equation}

Let
\[
 \begin{aligned}
  \vct X&:=(\vct X_1,\ldots,\vct X_m)\in\cL^m,\\
  \vct Y&:=(\vct Y_1,\ldots,\vct Y_m)\in\cL^m
 \end{aligned}
\]
be independent random arrays.  Within each array, the entries are
pairwise independent with marginal distribution \(\nu\).

Write \(\vct v^*=\mat B\vct z^*\), so
\(\norm{\vct v^*}\le \norm{\vct b_1}\).  Concatenate the two arrays as
\[
 \vct U:=(\vct X_1,\ldots,\vct X_m,
           \vct Y_1,\ldots,\vct Y_m).
\]
Define the \emph{target-pair count} \(Z\) and the second frequency
moment \(F_{\rm pt}\) by
\[
\begin{aligned}
 Z&:=\sum_{i,j\in[m]}\mathbf1[\vct X_i-\vct Y_j=\vct v^*],\\
 F_{\rm pt}&:=\Ftwo(\vct U)
 =\sum_{\vct x\in\cL}
 \bigl|\{t\in[2m]:\vct U_t=\vct x\}\bigr|^2.
\end{aligned}
\]
The quantity \(Z\) counts representations of \(\vct v^*\) as
\(\vct X_i-\vct Y_j\), while \(F_{\rm pt}\) sums the squared multiplicities of
lattice points in \(\vct U\).

\begin{theorem}
\label{thm:many-representations}
For all sufficiently large \(n\), the two arrays satisfy
\[
 Z\ge\frac{m^2\,2^{-20d}}{2\Pi}
\]
with probability at least \(1-2^{-\Omega(n)}\).
\end{theorem}

\begin{proof}
Set
\[
 z_0:=\frac{m^2\,2^{-20d}}{2\Pi}.
\]
Write \(\rho:=\rho_{\vct z^*}\).  By
\cref{thm:sampler-autocorrelation}, \(\mathbb E[Z]=m^2\rho\ge2z_0\).
For each pair of indices, define
\[
 \xi_{ij}:=\mathbf1[\vct X_i-\vct Y_j=\vct v^*].
\]
Then \(Z=\sum_{i,j\in[m]}\xi_{ij}\), and
\[
\begin{aligned}
 \operatorname{Var}[Z]
 &=\sum_{i,j\in[m]}\operatorname{Var}[\xi_{ij}]\\
 &\quad+\sum_{\substack{(i,j),(i',j')\in[m]^2\\(i,j)\ne(i',j')}}
   \operatorname{Cov}(\xi_{ij},\xi_{i'j'}).
\end{aligned}
\]
The first sum contains the \(m^2\) terms for which the two index pairs
coincide.  Since \(\xi_{ij}\) is Bernoulli with mean \(\rho\),
\[
 \sum_{i,j\in[m]}\operatorname{Var}[\xi_{ij}]
 =m^2\rho(1-\rho)\le m^2\rho.
\]
We divide the covariance terms in the second sum according to their
endpoints.

\emph{Case 1: disjoint endpoints.}
If \(i\ne i'\) and \(j\ne j'\), then \(\vct X_i,\vct X_{i'}\) are independent,
as are \(\vct Y_j,\vct Y_{j'}\).  The two arrays are independent, so these four
variables are jointly independent.  Hence
\(\xi_{ij}\) and \(\xi_{i'j'}\) are independent, and their covariance
is zero.

\emph{Case 2: one shared endpoint.}
If \(i=i'\) and \(j\ne j'\), then
\[
 \mathbb E[\xi_{ij}\xi_{ij'}]
 =\sum_{\vct x\in\cL}\nu(\vct x)\nu(\vct x-\vct v^*)^2
 \le\|\nu\|_\infty\rho.
\]
The case \(j=j'\) and \(i\ne i'\) is symmetric and gives the same
bound.
In either case,
\(\operatorname{Cov}(\xi_{ij},\xi_{i'j'})\le
\|\nu\|_\infty\rho\).  There are \(2m^2(m-1)\) ordered pairs of
indicators sharing exactly one endpoint.  Therefore
\[
 \operatorname{Var}[Z]
 \le m^2\rho+2m^2(m-1)\|\nu\|_\infty\rho
 \le m^2\rho+2m^3\|\nu\|_\infty\rho.
\]
Chebyshev's inequality gives
\[
\begin{aligned}
 \Pr[Z<\tfrac12\mathbb E[Z]]
 &\le\frac{4\operatorname{Var}[Z]}{\mathbb E[Z]^2}\\
 &\le4\left(\frac1{m^2\rho}
       +\frac{2\|\nu\|_\infty}{m\rho}\right)
 \le2^{-\Omega(n)},
\end{aligned}
\]
using \(m/\Pi=2^{64n}\), \(d\le n\), and the bounds on \(\rho\) and
\(\|\nu\|_\infty\) from \cref{thm:sampler-autocorrelation}.
Since \(z_0\le\tfrac12\mathbb E[Z]\), this proves the claim.
\end{proof}

We next bound \(F_{\rm pt}\), the second frequency moment of the sample
points.

\begin{theorem}
\label{thm:point-frequency}
The two arrays satisfy
\[
 F_{\rm pt}
 \le16\left(2m+\frac{4m^2\,2^{3d}}{\Pi}\right)
\]
with probability at least \(15/16\).
\end{theorem}

\begin{proof}
Set
\[
 f_0:=16\left(2m+\frac{4m^2\,2^{3d}}{\Pi}\right).
\]
By definition,
\[
 F_{\rm pt}
 =\sum_{s,t\in[2m]}\mathbf1[\vct U_s=\vct U_t].
\]
The \(2m\) terms
with \(s=t\) contribute \(2m\).  For \(s\ne t\), pairwise independence
gives
\[
 \Pr[\vct U_s=\vct U_t]
 =\sum_{\vct x\in\cL}\nu(\vct x)^2
 \le\|\nu\|_\infty.
\]
There are \(2m(2m-1)\) ordered pairs with \(s\ne t\).  Hence
\[
\begin{aligned}
 \mathbb E[F_{\rm pt}]
 &=2m+\sum_{\substack{s,t\in[2m]\\s\ne t}}\Pr[\vct U_s=\vct U_t]\\
 &\le2m+2m(2m-1)\|\nu\|_\infty\\
 &\le2m+\frac{4m^2\,2^{3d}}{\Pi}=\frac{f_0}{16}.
\end{aligned}
\]
Markov's inequality gives
\(\Pr[F_{\rm pt}>f_0]\le1/16\), which proves the claim.
\end{proof}

By a union bound, \cref{thm:many-representations,thm:point-frequency}
imply that, for all sufficiently large \(n\), \(Z\ge z_0\) and
\(F_{\rm pt}\le f_0\) hold simultaneously with probability at least
\[
 1-2^{-\Omega(n)}-1/16\ge3/4.
\]

Generating a lattice point according to \(\nu\) uses an \(\ell\)-bit random
string.  To generate the two arrays from short seeds, set
\[
 \bar\ell:=\max\{\ell,\lceil\log_2m\rceil\}
\]
and fix an ordered \(\F_2\)-basis of \(\F_{2^{\bar\ell}}\).
For \(j\in[m]\), let \(t_j\)
be represented by the \(\bar\ell\)-bit encoding of \(j-1\), so \(t_j\)
is computed from \(j\) rather than stored.  For each
\(\vct A\in\{\vct X,\vct Y\}\), choose an independent random affine
polynomial
\[
 h_{\vct A}(t)=a_{\vct A}t+b_{\vct A},
 \qquad a_{\vct A},b_{\vct A}\in\F_{2^{\bar\ell}},
\]
with independent uniform coefficients.  The first \(\ell\)
coordinates of \(h_{\vct A}(t_j)\) provide the random bits used to generate
\(\vct A_j\).  These bit strings are pairwise independent and uniform, so
the entries within each array are pairwise independent with marginal
distribution \(\nu\).  Storing the
two coefficients uses \(2\bar\ell\) bits per array and makes repeated
queries consistent~\cite{CarterWegman79}.  The field construction and
evaluations take deterministic polynomial time and space~\cite{Shoup90}.

\subsection{From short differences to collisions}
\label{subsec:grid-collisions}
The arrays \(\vct X\) and \(\vct Y\) now contain many pairs
with difference \(\vct v^*\), but the collision search detects equal
labels rather than lattice differences.  We bridge this gap by assigning
a label to each array entry in three steps.  A
random signed Hadamard transform spreads a short difference across the
coordinates, and a randomly shifted grid assigns the two lattice
points defining that difference to the same cell with probability at least
\(\gamma=2^{-O(n)}\).  Every
cell has bounded diameter, so a lattice packing bound limits the number
of distinct lattice points with the same cell label.  Combined with
the bound on \(F_{\rm pt}\) from \cref{subsec:representations}, this
controls the second frequency moment of the cell-label
array.  Finally, a pairwise-universal hash maps the cell labels to the
range required by the collision search.  \Cref{thm:grid-label-bounds}
combines these facts: many target pairs receive the same label, while the
second frequency moment of the label array remains small.

We first define the random grid.  Set
\[
 \bar n:=2^{\lceil\log_2n\rceil}\le2n,
\]
and identify \(\Z^n\) with its embedding in \(\Z^{\bar n}\) obtained by
appending \(\bar n-n\) zero coordinates.  Let \(\mat H_{\bar n}\) be the
Sylvester Hadamard matrix.  Choose independent uniform signs
\(\varepsilon_j\in\{-1,1\}\), and set
\[
 \mat D_{\vct{\varepsilon}}
 :=\operatorname{diag}(\varepsilon_1,\ldots,\varepsilon_{\bar n}),
 \qquad
 \mat S:=\mat H_{\bar n}\mat D_{\vct{\varepsilon}}.
\]
Since \(\mat S^{\mathsf T}\mat S=\bar n\mat I_{\bar n}\), every vector
\(\vct v\) with \(\norm{\vct v}\le\norm{\vct b_1}\) satisfies
\[
 \norm{\mat S\vct v}=\sqrt{\bar n}\norm{\vct v}
 \le\sqrt{\bar n}\norm{\vct b_1}.
\]

Choose the least power of two \(W\) satisfying
\begin{equation}\label{eq:grid-width}
 W^2\ge128\norm{\vct b_1}^2\log_2\bar n.
\end{equation}
This condition can be checked exactly because
\(\log_2\bar n=\lceil\log_2n\rceil\) is an integer.
Independently choose a shift \(\vct{\sigma}=(\sigma_1,\ldots,\sigma_{\bar n})\)
whose coordinates are independent and uniform on \(\{0,\ldots,W-1\}\).
Write
\(g:=(\mat D_{\vct{\varepsilon}},\vct{\sigma})\) for the resulting random
grid.  Define the \emph{cell label} of \(\vct x\in\Z^{\bar n}\) as the
integer tuple \(\Cell_g(\vct x)\in\Z^{\bar n}\) with
\[
 \Cell_g(\vct x)_j:=
 \left\lfloor\frac{(\mat S\vct x)_j+\sigma_j}{W}\right\rfloor,
 \qquad j\in[\bar n].
\]
The level sets of \(\Cell_g\) form the cells of the random grid.

The next lemma bounds the collision probability for \(\norm{\vct v}\le \norm{\vct b_1}\)
and the cell diameter.

\begin{lemma}\label{lem:grid-collision}
Let \(\vct x,\vct v\in\Z^{\bar n}\) satisfy
\(\norm{\vct v}\le\norm{\vct b_1}\), and set
\[
 \gamma:=\frac12(3/4)^{\bar n}=2^{-O(n)}.
\]
Then
\[
 \Pr_g[\Cell_g(\vct x)=\Cell_g(\vct x-\vct v)]\ge\gamma.
\]
Moreover, any \(\vct x,\vct y\in\Z^{\bar n}\) in the same cell satisfy
\[
 \norm{\vct x-\vct y}<W.
\]
\end{lemma}

\begin{proof}
The collision claim is immediate when \(\vct v=0\), so assume
\(\vct v\ne0\).  For fixed \(j\in[\bar n]\),
\[
 (\mat S\vct v)_j=\sum_{k=1}^{\bar n}
 (\mat H_{\bar n})_{j,k}\varepsilon_kv_k.
\]
The summands are independent and have mean zero.  Since every Hadamard
entry has absolute value one,
\[
\begin{aligned}
 \mathbb E_{\mat D_{\vct{\varepsilon}}}[(\mat S\vct v)_j]&=0,\\
 \operatorname{Var}_{\mat D_{\vct{\varepsilon}}}[(\mat S\vct v)_j]
 &=\sum_{k=1}^{\bar n}v_k^2
 =\norm{\vct v}^2\le \norm{\vct b_1}^2.
\end{aligned}
\]
Hoeffding's inequality therefore gives
\[
\begin{aligned}
 \Pr_{\mat D_{\vct{\varepsilon}}}\left[|(\mat S\vct v)_j|>\frac W4\right]
 &\le2\exp\left(-\frac{W^2}{32\norm{\vct v}^2}\right)\\
 &\le2\exp\left(-\frac{W^2}{32\norm{\vct b_1}^2}\right).
\end{aligned}
\]
A union bound and Equation~\eqref{eq:grid-width} give, since
\(\bar n\ge2\),
\[
\begin{aligned}
 \Pr_{\mat D_{\vct{\varepsilon}}}\left[
 \max_{j\in[\bar n]}|(\mat S\vct v)_j|>\frac W4\right]
 &\le2\bar n\exp\left(-\frac{W^2}{32\norm{\vct b_1}^2}\right)\\
 &\le2\bar n\exp(-4\log_2\bar n)
 \le4e^{-4}<\frac12.
\end{aligned}
\]
Thus, with probability at least \(1/2\),
\[
 \max_{j\in[\bar n]}|(\mat S\vct v)_j|\le W/4.
\]
Fix a sign matrix \(\mat D_{\vct{\varepsilon}}\) satisfying this bound.
For each \(j\), set
\[
 a_j:=(\mat S\vct x)_j,\qquad
 b_j:=(\mat S(\vct x-\vct v))_j.
\]
Both quantities are integers.  Moreover,
\[
 |a_j-b_j|=|(\mat S\vct v)_j|\le W/4.
\]
Exactly
\(|a_j-b_j|\) integers \(\sigma_j\in\{0,\ldots,W-1\}\) make
\(a_j+\sigma_j\) and \(b_j+\sigma_j\) belong to different intervals
\([kW,(k+1)W)\), where \(k\in\Z\).  Hence
\[
 \Pr_{\sigma_j}\left[
 \left\lfloor\frac{a_j+\sigma_j}{W}\right\rfloor
 =\left\lfloor\frac{b_j+\sigma_j}{W}\right\rfloor
 \right]
 =1-\frac{|a_j-b_j|}{W}\ge\frac34.
\]
The shifts are independent, so
\[
 \Pr_{\vct{\sigma}}[\Cell_g(\vct x)=\Cell_g(\vct x-\vct v)
 \mid \mat D_{\vct{\varepsilon}}]
 \ge(3/4)^{\bar n}.
\]
Averaging over \(\mat D_{\vct{\varepsilon}}\) proves the collision bound.

If \(\vct x,\vct y\) share a cell, each coordinate of \(\mat S(\vct x-\vct y)\) has absolute value
less than \(W\).  Hence
\[
 \sqrt{\bar n}\norm{\vct x-\vct y}
 =\norm{\mat S(\vct x-\vct y)}<\sqrt{\bar n}W.
\]
Dividing by \(\sqrt{\bar n}\) gives
\(\norm{\vct x-\vct y}<W\).
\end{proof}

The following \cref{lem:cell-occupancy} gives an upper bound for the
number of lattice points in each cell.

\begin{lemma}\label{lem:cell-occupancy}
Given a quasi-HKZ basis with rank $d$,
every cell contains at most
\[
 K:=\left(1+128\left\lceil\sqrt{\log_2n}\right\rceil\right)^d
 =\exp(O(d\log\log n))
\]
distinct points in \(\cL\).
\end{lemma}

\begin{proof}
Let \(\mathcal Q\) be the set of lattice points in a fixed cell.
The claim is immediate if \(\mathcal Q=\varnothing\), so fix
\(\vct x_0\in\mathcal Q\).  By
\cref{lem:grid-collision}, every \(\vct x\in\mathcal Q\) satisfies
\(\norm{\vct x-\vct x_0}<W\).  Moreover, the open balls of radius
\(\lambda_1/2\) centered at the points of \(\mathcal Q\) are pairwise disjoint,
and all lie in the ball of radius \(W+\lambda_1/2\) centered at
\(\vct x_0\).  Comparing \(d\)-dimensional volumes in the lattice span gives
\[
 |\mathcal Q|\le
 \left(\frac{W+\lambda_1/2}{\lambda_1/2}\right)^d
 =\left(1+\frac{2W}{\lambda_1}\right)^d.
\]
The minimality of \(W\) in Equation~\eqref{eq:grid-width} and
\(\log_2\bar n=\lceil\log_2n\rceil\le2\log_2n\) give
\[
 W<2\sqrt{128}\,\norm{\vct b_1}\sqrt{\log_2\bar n}
 \le32\norm{\vct b_1}\sqrt{\log_2n}.
\]
Together with \(\norm{\vct b_1}\le2\lambda_1\) from
\cref{lem:hkz-shape}, this implies
\[
 \frac{2W}{\lambda_1}<128\sqrt{\log_2n}
 \le128\left\lceil\sqrt{\log_2n}\right\rceil.
\]
Substitution proves \(|\mathcal Q|\le K\).  Finally,
\[
 \left\lceil\sqrt{\log_2n}\right\rceil
 =\left\lceil\sqrt{\lceil\log_2n\rceil}\right\rceil,
\]
so \(K\) is exactly computable.
\end{proof}

Fix the arrays \(\vct X\) and \(\vct Y\), and fix a grid \(g\).
Let \(n_{\vct x}\) be the multiplicity of \(\vct x\) in \(\vct U\).  Since each cell
\(C\) contains at most \(K\) distinct points, applying Cauchy--Schwarz
within each cell and then summing over the cells gives
\begin{equation}\label{eq:cell-frequency}
  F_{\rm cell}:=
  \sum_C\left(\sum_{\vct x\in C}n_{\vct x}\right)^2\le K\sum_{\vct x}n_{\vct x}^2=K F_{\rm pt}.
\end{equation}

To apply \cref{lem:collision-search}, represent each of the \(2m\)
cell labels by a canonical bit string of a common polynomial length.  To
compress these labels to a universe of size \((2m)^{O(1)}\), choose \(h\)
independently and uniformly from a pairwise-universal family of binary
affine maps to \(\{0,1\}^{q_h}\), where
\(q_h:=\lceil8\log_2(2m)\rceil\).  This range has size less than
\(2(2m)^8\).

Define
\[
 Z_g:=\sum_{\substack{i,j\in[m]\\\vct X_i-\vct Y_j=\vct v^*}}
 \mathbf1[\Cell_g(\vct X_i)=\Cell_g(\vct Y_j)].
\]
Apply \(h\) to the cell labels in the order of \(\vct U\) to
obtain the compressed label array
\[
 \vct C_{g,h}:=\bigl(h(\Cell_g(\vct U_t))\bigr)_{t\in[2m]}.
\]
The expected number over \(h\) of ordered pairs
with unequal cell labels but equal compressed labels is at most
\[
 \frac{(2m)^2}{2^{q_h}}\le(2m)^{-6}.
\]
Equal cell labels always
have equal compressed labels.  Hence, by Markov's inequality, \(h\) creates no
equality between unequal cell labels with probability at least
\(1-(2m)^{-6}\ge63/64\).

\begin{theorem}
\label{thm:grid-label-bounds}
Assume
\[
 Z\ge z_0,\qquad F_{\rm pt}\le f_0,
\]
and set
\[
 r_0:=\left\lfloor\frac{\gamma z_0}{4}\right\rfloor,\qquad
 p_0:=\min\{(2m)^2,Kf_0\}.
\]
If \(g\) and \(h\) are sampled independently as above, then
\[
 \Pr_{g,h}\!\left[
   \Ftwo(\vct C_{g,h})\le p_0
   \ \text{and}\ Z_g\ge r_0
 \right]\ge\frac{\gamma}{8}.
\]
\end{theorem}

\begin{proof}
Since \(m=2^{64n}\Pi\), \(d\le n\),
\(\bar n\le2n\), and \(\Pi\ge1\), the definitions give
\[
 \frac{\gamma z_0}{4}
 \ge2^{108n-4}(3/4)^{2n}\longrightarrow\infty.
\]
Hence \(r_0\ge1\)
for all sufficiently large \(n\).

For each target pair, its two lattice points are assigned to the same
cell with probability at least \(\gamma\) by \cref{lem:grid-collision}, so
\(\mathbb E_g[Z_g]\ge\gamma Z\).  Since \(0\le Z_g\le Z\), we also have
\(\mathbb E_g[Z_g^2]\le Z\mathbb E_g[Z_g]\).  Paley--Zygmund now gives
\[
 \Pr_g[Z_g\ge\tfrac12\mathbb E_g[Z_g]]
 \ge\frac{(\mathbb E_g[Z_g])^2}{4\mathbb E_g[Z_g^2]}
 \ge\frac{\mathbb E_g[Z_g]}{4Z}
 \ge\frac\gamma4.
\]
Let \(\mathcal A:=\{Z_g\ge\tfrac12\mathbb E_g[Z_g]\}\).  Whenever
\(\mathcal A\) occurs,
\[
 Z_g\ge\gamma Z/2\ge\gamma z_0/2\ge r_0.
\]

The cell multiplicities sum to \(2m\), so
\(F_{\rm cell}\le(2m)^2\).  Equation~\eqref{eq:cell-frequency} gives
\(F_{\rm cell}\le Kf_0\).
Thus \(F_{\rm cell}\le\min\{(2m)^2,Kf_0\}=p_0\) for every \(g\).

Let \(\mathcal H_g\) be the event that \(h\) is injective on
the distinct cell labels represented in the two arrays.  The
compression bound above gives
\(\Pr_h[\mathcal H_g\mid g]\ge63/64\) for every \(g\).  If both
\(\mathcal A\) and \(\mathcal H_g\) occur, then
\[
 \Ftwo(\vct C_{g,h})=F_{\rm cell}\le p_0,
 \qquad Z_g\ge r_0.
\]
Therefore,
\[
 \Pr_{g,h}[\mathcal A\cap\mathcal H_g]
 \ge\frac{63}{64}\Pr_g[\mathcal A]
 \ge\frac{63\gamma}{256}\ge\frac\gamma8.
\]
\end{proof}

\section{The complete SVP algorithm}
\label{sec:solver}
We now use the sampler and the collision search in~\cref{sec:sampler} to solve SVP, which is formally stated in the following \cref{thm:main}.

\begin{theorem}\label{thm:main}
There is a classical algorithm that, given any full-rank integer basis
\(\mat B\in\Z^{n\times n}\), returns a shortest nonzero vector of
\(\cL(\mat B)\) with probability at least \(2/3\).  On every execution,
it uses polynomial space and performs at most
\[
 \poly(|\mat B|)\,n^{\frac{n}{4e}+o(n)}
\]
bit operations, where \(|\mat B|\) is the input bit length.
\end{theorem}

We first solve the problem for a quasi-HKZ basis in~\cref{subsec:solver-quasihkz}.  By
\cref{thm:many-representations}, the sample arrays contain many pairs whose
difference is a shortest vector. The random grid in \cref{thm:grid-label-bounds} turns many of these pairs into collisions while keeping the second frequency moment small. We can therefore apply \cref{lem:collision-search} to find one
such pair.  We then use this solver in Kannan's recursion to handle an
arbitrary basis in~\cref{subsec:complete-solver}.

\subsection{Solving SVP on a quasi-HKZ basis}
\label{subsec:solver-quasihkz}
Hanrot and Stehl\'e show that Kannan's method enumerates all lattice
vectors in a prescribed ball, which satisfies the shape bounds~\cref{lem:hkz-shape}.  Since we
need only one shortest vector, we replace this enumeration by collision
search.  Our subroutine samples two implicit arrays once and then repeats
the grid search.  On each grid, it searches for equal cell labels,
checks the corresponding difference exactly, and stores the shortest
valid difference.

\begin{figure}[htbp]
\centering
\begingroup
\setlength{\fboxsep}{7pt}
\setlength{\fboxrule}{0.4pt}
\fbox{%
\begin{minipage}{\dimexpr\linewidth-2\fboxsep-2\fboxrule\relax}
\small

\noindent\textbf{Input:}
A quasi-HKZ basis
\(\mat B=[\vct b_1\ \cdots\ \vct b_d]\in\Z^{n\times d}\)
of rank \(d\le n\).

\noindent\textbf{Output:}
A nonzero vector \(\vct v\in\cL(\mat B)\) satisfying
\(\norm{\vct v}\le\norm{\vct b_1}\).  With probability at least
\(2/3\), it is a shortest vector of \(\cL(\mat B)\).

\begin{enumerate}[
  label=\arabic*.,
  leftmargin=2.2em,
  labelsep=0.5em,
  itemsep=0.25\baselineskip,
  topsep=0.45\baselineskip,
  parsep=0pt
]
\item Set up the sampler for \(\nu\) from \cref{sec:sampler} and use the array length
      \(m\) from Equation~\eqref{eq:implicit-array-length}.  Let \(\gamma\)
      be the lower bound from \cref{lem:grid-collision}.  Two lattice points
      at distance at most \(\norm{\vct b_1}\) share a grid cell with probability at
      least \(\gamma\).

\item Draw seeds for the two implicit arrays \(\vct X,\vct Y\).
      Set \(\vct v\gets\vct b_1\).

\item Repeat a fixed, sufficiently large constant multiple of \(1/\gamma\)
      times:
  \begin{enumerate}[
    label=(\alph*),
    leftmargin=1.8em,
    itemsep=0.15\baselineskip,
    topsep=0.2\baselineskip,
    parsep=0pt
  ]
  \item Draw a fresh grid \(g\) and compression map \(h\).  Use the numerical bounds
        \(p_0,r_0\) from \cref{thm:grid-label-bounds} in the search of
        \cref{lem:collision-search} on the combined compressed label
        array.
  \item For every returned cross-array pair \((\vct X_i,\vct Y_j)\),
        recompute the two samples.  Accept
        the pair only if
        \[
          \Cell_g(\vct X_i)=\Cell_g(\vct Y_j)
          \qquad\text{and}\qquad
          0<\norm{\vct X_i-\vct Y_j}\le\norm{\vct b_1}.
        \]
  \item If an accepted pair gives a shorter vector than \(\vct v\), replace
        \(\vct v\) by its difference.
  \end{enumerate}
\item Return \(\vct v\).
\end{enumerate}

\smallskip
All array, grid, compression, and walk seeds used above are mutually
independent.

\end{minipage}%
}
\endgroup
\caption{The SVP solver for a quasi-HKZ basis.}
\label{fig:quasi-hkz-solver}
\end{figure}

The next theorem gives the correctness, success probability, and cost of
this subroutine.

\begin{theorem}\label{thm:quasi-hkz-solver}
Given a quasi-HKZ basis
\(\mat B\in\Z^{n\times d}\) of rank \(d\le n\), the solver always returns
a vector \(\vct v\in\cL(\mat B)\setminus\{\vct 0\}\) satisfying
\(\norm{\vct v}\le\norm{\vct b_1}\).  With probability at least
\(2/3\), this vector is shortest.  The solver uses polynomial space and
performs at most
\[
 2^{O(n)}\exp(O(d\log\log n))\,
 d^{d/(4e)+o(d)}\poly(|\mat B|)
\]
bit operations.
\end{theorem}

\begin{proof}
The initial vector is \(\vct b_1\).  Every sampled point belongs to
\(\cL(\mat B)\), so every stored difference does as well.  The exact test
excludes zero and requires norm at most \(\norm{\vct b_1}\), so the output is always a
nonzero lattice vector no longer than \(\vct b_1\).

\mypara{Success probability}
Fix a shortest vector \(\vct v^*\).  Since \(\vct b_1\) is a nonzero
lattice vector, \(\norm{\vct v^*}\le\norm{\vct b_1}\).
By \cref{thm:many-representations,thm:point-frequency} and a union bound,
for all sufficiently large \(n\), with probability at least \(3/4\) the two
arrays contain at least \(z_0\) pairs whose difference is \(\vct v^*\),
and their second frequency moment is at most \(f_0\).  Fix arrays with
these two properties.

For a fresh grid and compression map, \cref{thm:grid-label-bounds} gives the
following with probability at least \(\gamma/8\): the compressed label
array has second frequency moment at most \(p_0\) and contains at least
\(r_0\) target pairs.  Conditioned on this event, the two hypotheses of
\cref{lem:collision-search} hold with \(p=p_0\) and \(\cG\) equal to the set
of target pairs.  Hence the collision search returns a target pair with
constant probability.  The
pair passes the exact test:
its points share a cell, and their difference is the nonzero vector
\(\vct v^*\) of norm at most \(\norm{\vct b_1}\).

Thus one outer iteration succeeds with probability \(\Omega(\gamma)\).
The iterations use independent seeds, so a sufficiently large constant
multiple of \(1/\gamma\) iterations has failure probability at most
\(1/9\).  Including the probability that the arrays have the two properties
above, the success probability is at least \((3/4)(8/9)=2/3\).
When a target pair is inspected, either the current candidate is already
shortest or it is replaced by \(\vct v^*\).
No later accepted vector can be shorter.

\mypara{Complexity}
We first bound the number of array accesses.
The parameters in \cref{thm:grid-label-bounds} satisfy
\(r_0\ge\gamma z_0/8\) and \(p_0\le Kf_0\).  The definitions of
\(z_0,f_0\) in \cref{thm:many-representations,thm:point-frequency}, together with
\(m=2^{64n}\Pi\) from Equation~\eqref{eq:implicit-array-length}, give
\(f_0/z_0=2^{O(d)}\).  It follows that
\[
 \frac{p_0}{r_0}\le \gamma^{-1}K2^{O(d)}.
\]
The access bound in \cref{lem:collision-search}, applied to an array of
length \(2m\), is therefore
\[
 \widetilde O\!\left(
   \sqrt m\,\gamma^{-1}K2^{O(d)}
 \right)
\]
per outer iteration.  We use the same fixed access limit even when the
grid does not satisfy the lemma's hypotheses.  Thus the bound holds for
every choice of the seeds.  The \(O(\gamma^{-1})\) outer iterations add the
second loss in \(\gamma\), for a total of
\[
 \widetilde O\!\left(
   \sqrt m\,\gamma^{-2}K2^{O(d)}
 \right)
\]
array accesses.

It remains to bound the three factors in this expression.
By \cref{lem:grid-collision}, \(\gamma^{-2}=2^{O(n)}\).
By \cref{lem:cell-occupancy},
\(K=\exp(O(d\log\log n))\).  Finally,
Equation~\eqref{eq:implicit-array-length} gives
\(\sqrt m=2^{32n}\sqrt\Pi\).  Equation~\eqref{eq:rounded-scales} gives
\(\Pi<2^dP\), while \cref{lem:hkz-shape} gives
\(P\le2^{O(d)}d^{d/(2e)}\).  Hence
\[
 \sqrt m\le2^{O(n)}d^{d/(4e)}.
\]

By \cref{lem:coordinate-sampler}, sampling all \(d\) coordinates takes
\(d^{o(d)}\poly(|\mat B|)\) time and polynomial space.  Evaluating the
seed polynomial and computing the cell label add only polynomial work.
Each walk reports at most one pair,
so the access bound also covers all exact checks.  The total number of bit
operations is therefore
\[
 2^{O(n)}\exp(O(d\log\log n))\,
 d^{d/(4e)+o(d)}\poly(|\mat B|).
\]

The arrays are stored by their seeds.  The space bound in
\cref{lem:coordinate-sampler} covers sample evaluation, and
\cref{lem:collision-search} uses only \(O(\log^3(2m))\) additional bits for
the walk.  The bounds above give \(\log m=O(n+d\log d)\).  Thus the grid,
compression map, current samples, and best vector all have
polynomial-size descriptions.  This proves the claimed space bound.

\end{proof}

\subsection{Kannan's recursion for arbitrary bases}
\label{subsec:complete-solver}
To handle an arbitrary basis, we use Kannan's
recursion~\cite{Kannan87,HanrotStehle07}.
It first recursively reduces the projected basis
\([\pi_2(\vct b_2)\ \cdots\ \pi_2(\vct b_d)]\) of
\(\pi_2(\cL(\mat B))\) to an HKZ-reduced basis.  Whenever the norm condition
\(\norm{\vct b_2^*}\ge\norm{\vct b_1}/2\) fails, Gauss reduction
replaces \(\vct b_1\) by a shorter vector, and the process repeats.
Once the basis is quasi-HKZ,
the solver in \cref{thm:quasi-hkz-solver} finds a shortest vector with
probability at least \(2/3\).  We run it \(n^2\) times and keep the shortest
result.  We put this vector first and recursively reduce the
corresponding projected basis once more.  When every group of
solver runs contains a shortest vector, the resulting basis is HKZ reduced.

\Cref{fig:complete-solver} gives the recursive procedure
\textsc{Reduce}.  At the top level, we save \(\vct b_1\), call
\textsc{Reduce}\((\mat B)\), and return the first vector of the
result.  If the call exceeds the time or space bound proved below, we
stop and return the saved vector.

\begin{figure}[!htbp]
\centering
\begingroup
\setlength{\fboxsep}{5pt}
\setlength{\fboxrule}{0.4pt}
\fbox{%
\begin{minipage}{\dimexpr\linewidth-2\fboxsep-2\fboxrule\relax}
\small

\noindent\textbf{Input:} A rank-\(d\) integer lattice basis
\(\mat B=[\vct b_1\ \cdots\ \vct b_d]\in\Z^{n\times d}\).

\noindent\textbf{Output:} A basis of the same lattice, HKZ reduced if each
group of solver runs finds a shortest vector.

\begin{enumerate}[
  label=\arabic*.,
  leftmargin=1.8em,
  labelsep=0.35em,
  itemsep=0pt,
  topsep=0.2\baselineskip,
  parsep=0pt
]
\item LLL-reduce \(\mat B\).
\item If \(d=1\), return \(\mat B\).
\item Repeat until:
\item \hspace*{1em}Compute the projections of
      \(\vct b_2,\ldots,\vct b_d\) orthogonally to \(\vct b_1\).
\item \hspace*{1em}Apply \textsc{Reduce} to this projected basis.
\item \hspace*{1em}Lift the result to lattice vectors and size-reduce
      them against \(\vct b_1\).
\item \hspace*{1em}If \(\norm{\vct b_2^*}\ge\norm{\vct b_1}/2\),
      leave the loop.
\item \hspace*{1em}Gauss-reduce \((\vct b_1,\vct b_2)\) and repeat.
\item Run the solver in \cref{fig:quasi-hkz-solver} \(n^2\) times.
      Divide each output by the gcd of its coefficients in \(\mat B\).
      Let \(\vct v\) be the shortest resulting lattice vector.
\item Set \(\mat B\gets\operatorname{LLL}
      (\vct v,\vct b_1,\ldots,\vct b_d)\), processing
      \(\vct v\) first and deleting dependent vectors.
\item Compute the projections of \(\vct b_2,\ldots,\vct b_d\)
      orthogonally to \(\vct b_1\).
\item Apply \textsc{Reduce} to this projected basis.
\item Lift the result to lattice vectors, size-reduce them against
      \(\vct b_1\), and return \(\mat B\).
\end{enumerate}

\end{minipage}%
}
\endgroup
\caption{Kannan's recursive reduction \textsc{Reduce}.}
\label{fig:complete-solver}
\end{figure}
\FloatBarrier

The projected bases have rational entries.  Before each recursive call,
we clear their denominators.  To lift the result, we record the integer
basis changes, apply them to the corresponding unprojected
vectors \(\vct b_2,\ldots,\vct b_d\), and size-reduce against
the first vector.  All solver calls use independent seeds.

The next lemma bounds the number of loop rounds.  It uses the standard
analysis of Kannan's recursion~\cite{Kannan87,LLL82,Helfrich85,HanrotStehle07}.

\begin{lemma}
\label{lem:kannan-preprocessing}
For \(d\ge2\), the loop in
\cref{fig:complete-solver} produces a quasi-HKZ basis in
\(O(\log d)\) rounds.
\end{lemma}

\begin{proof}
After recursively reducing the projected basis and lifting
the result, the basis is size reduced and its projected basis is HKZ
reduced.  By \cref{def:quasi-hkz}, only the norm test remains.  We show
that it can fail only \(O(\log d)\) times.

Write \(\lambda_1:=\lambda_1(\cL(\mat B))\).  If the test fails, size reduction
gives \(\norm{\vct b_2}<\norm{\vct b_1}/\sqrt2\).  Thus Gauss reduction replaces
the first vector by \(\vct b_1'\) with
\(\norm{\vct b_1'}<\norm{\vct b_1}/\sqrt2\).
A shortest vector cannot be parallel to \(\vct b_1\), because
every lattice vector in that direction is an integer multiple of
\(\vct b_1\), whereas \(\vct b_2\) is shorter.  Consequently, the
projection of every shortest vector is nonzero, so the HKZ-reduced
projected basis gives \(\norm{\vct b_2^*}\le\lambda_1\).
 The two-dimensional area bound for Gauss reduction now gives
 \[
 \begin{aligned}
  \norm{\vct b_1'}^2&\le2\norm{\vct b_1}\norm{\vct b_2^*}
  \le2\norm{\vct b_1}\lambda_1,\\
 \frac{\norm{\vct b_1'}}{\lambda_1}
 &\le\sqrt{\frac{2\norm{\vct b_1}}{\lambda_1}}.
\end{aligned}
\]
LLL starts with \(\norm{\vct b_1}/\lambda_1\le2^{(d-1)/2}\).  Repeated square roots make
this ratio bounded by a constant after \(O(\log d)\) failed rounds.
Each subsequent failure reduces \(\norm{\vct b_1}\) by a factor of at least \(\sqrt2\),
so only constantly many more are possible before \(\norm{\vct b_1}<\lambda_1\).
Since \(\vct b_1\) is always a nonzero lattice vector, this cannot
happen, and the loop must stop.

The basis changes preserve the lattice.  The usual exact implementation
of LLL and recursive lifting keeps the intermediate bases and recorded
changes polynomial in the original input size.
\end{proof}

\begin{proof}[of \cref{thm:main}]
We first show that the recursion gives an HKZ-reduced basis when it
finds a shortest vector at every call.  We then bound the probability
that a search fails and the total cost.

\mypara{Correctness}
Suppose that, at every recursive call, at least one of the \(n^2\) solver
runs returns a shortest vector.  We prove by induction on the rank that
\textsc{Reduce} returns an HKZ-reduced basis of the same lattice.
For \(d=1\), the returned rank-one basis is HKZ reduced.
Assume henceforth that \(d\ge2\).
By induction, the recursive calls in the loop return HKZ-reduced
projected bases.  Thus \cref{lem:kannan-preprocessing} shows that the
loop stops after \(O(\log d)\) rounds with a quasi-HKZ basis.

By \cref{thm:quasi-hkz-solver}, each solver output is a nonzero lattice
vector.  Dividing its coefficients by their gcd keeps it nonzero and in
the lattice and cannot increase its length.  At least one solver output is
shortest and is unchanged by gcd division, since a nontrivial common divisor
would yield a shorter nonzero lattice vector.  Every resulting vector has
norm at least that of this shortest output.  Hence the shortest resulting
vector \(\vct v\) is shortest.  Its coefficients have gcd one, so it can be
extended to a basis.

LLL preserves the lattice because the input list still contains the old
basis.  It also keeps \(\vct v\) first: no nonzero lattice vector is
shorter, so the LLL test never swaps it out of the first position.  The
last recursive call reduces the new projected basis, and lifting makes
the full basis size reduced.  By \cref{def:quasi-hkz}, a shortest first
vector and an HKZ-reduced projected basis give an HKZ-reduced basis.
This proves the induction and the correctness of the top-level output.

\mypara{Success probability}
Consider the next group of \(n^2\) solver runs.  Fix the full history
before these runs and suppose that every earlier group found a shortest
vector.  The argument above shows that the current input is quasi-HKZ.
For this fixed input, \cref{thm:quasi-hkz-solver} bounds the failure
probability of one run by \(1/3\).  The runs use independent seeds, so
the probability that all \(n^2\) fail is at most \(3^{-n^2}\).

Before the first such failure, \cref{lem:kannan-preprocessing} bounds the
number of loop rounds by \(O(\log d)\) at rank \(d\).  There is one
lower-rank call per round and one more after inserting \(\vct v\).
The number of recursive calls is therefore at most
\[
 \prod_{d=2}^n O(\log(d+1))=\exp(O(n\log\log n)).
\]
A union bound over the possible first failure bounds the probability of
any such failure by
\[
 \exp(O(n\log\log n))\,3^{-n^2}
 =2^{-\Omega(n^2)}<\frac13.
\]
Thus all the chosen vectors are shortest with probability at least
\(2/3\).

\mypara{Complexity}
First suppose that every group finds a shortest vector.  Here
\(|\mat B|\) denotes the original input length.  As noted in the proof of
\cref{lem:kannan-preprocessing}, the intermediate bases have bit length
polynomial in \(|\mat B|\).  By \cref{thm:quasi-hkz-solver}, each run
at rank \(d\) therefore takes
\[
 2^{O(n)}\exp(O(d\log\log n))\,
 d^{d/(4e)+o(d)}\poly(|\mat B|)
\]
bit operations.  The \(n^2\) repetitions add only a polynomial factor.
All other work at each node is polynomial.  The recursion has
\(\exp(O(n\log\log n))=n^{o(n)}\) nodes, and \(d\le n\), so the total
time is
\[
 n^{n/(4e)+o(n)}\poly(|\mat B|).
\]
The solver uses polynomial space by \cref{thm:quasi-hkz-solver}.
We run the calls one after another,
and each recursive call decreases the rank.  At most \(n\) levels are
active at once, so the total space is polynomial.

Choose the top-level time and space limits to cover these bounds.  They
are never reached when every group finds a shortest vector.  On any
other execution they bound the cost, and stopping returns the saved
nonzero vector \(\vct b_1\).  The success probability remains at
least \(2/3\), which proves the theorem.
\end{proof}

\section{Conclusion and Future Work}
\label{sec:conclusion}

We presented a randomized polynomial-space algorithm for exact SVP with running time \(n^{\frac{n}{4e}+o(n)}\). The main idea is to replace the exhaustive enumeration step in Kannan's algorithm~\cite{Kannan87,HanrotStehle07} with the low-space collision search of Lyu and Zhu~\cite{LyuZhu23}. For this approach to work, we design a sampling distribution so that a shortest vector can be written in many ways as the difference of two samples, while the number of irrelevant collisions stays under control.

This work leaves two natural directions for future research. The first is to improve the running time while keeping the space polynomial. Possible approaches include better sampling distributions, faster low-space collision methods, and a reduction in the loss caused by the random grid. It would also be interesting to obtain a similar improvement with a deterministic algorithm.

The second direction is to study whether these ideas can be used in practice, for example within the General Sieve Kernel (G6K) framework and implementation~\cite{G6K19}. The current algorithm is not meant for direct implementation. In particular, the array length \(m=2^{64n}\Pi\) is very large, and the random grid must be sampled \(O(1/\gamma)\) times, where \(\gamma\) is exponentially small. Moreover, each array entry is computed only when it is needed, so every access requires exact sampling, recomputation, and several hash evaluations. The quasi-HKZ preprocessing adds another large cost.

These costs come from parameter choices and tools used to prove the asymptotic bound, and a practical version would likely use a different setting. Possible changes include much smaller parameters, heuristic sampling and pruning, collision detection in batches, and reuse of the vector database already maintained by G6K~\cite{GamaNguyenRegev10,MicciancioWalter15,G6K19}. An interesting question is whether the main idea, i.e., giving a short vector many representations and then finding a useful collision, still works well after these changes.

\subsubsection*{AI disclosure.}
GPT-5.6 Sol was used interactively throughout the preparation of this manuscript, including early research, organization, exposition, and drafting. At an early stage, the model proposed a classical polynomial-space approach with running time \(n^{\frac{3n}{8e}+o(n)}\), based on the result of Chen et al.~\cite{ChenEtAl25}. Further exploration led to the present algorithm, which uses the low-space collision search of Lyu and Zhu~\cite{LyuZhu23}, improves the running time to \(n^{\frac{n}{4e}+o(n)}\), and gives a simpler analysis. The authors checked and substantially revised the final manuscript and remain solely responsible for all mathematical claims, proofs, references, and conclusions.

\bibliographystyle{splncs04}
\bibliography{references}

\end{document}